\documentclass[a4paper,onecolumn,11pt,accepted=2026-11-02]{quantumarticle}
\pdfoutput=1
\usepackage{graphicx} 
\usepackage[utf8]{inputenc}
\usepackage[english]{babel}
\usepackage[T1]{fontenc}
\usepackage{amsmath}
\usepackage{amsthm}
\usepackage{amsfonts}
\usepackage[numbers,sort&compress]{natbib}
\usepackage{hyperref}
\usepackage{mathrsfs}
\usepackage[linesnumbered,ruled,vlined]{algorithm2e}

\newtheorem{definition}{Definition}
\newtheorem{corollary}{Corollary}
\newtheorem{remark}{Remark}
\newtheorem{theorem}{Theorem}
\newtheorem{problem}{Problem}

\newcommand{\msqm}{\textsf{MS} quantum money}
\newcommand{\sea}{\textsf{SEA}}
\newcommand{\aicgadlog}{\textsf{AICGA\text{-}DLog}$_p$}
\newcommand{\qmkcgea}{$q$-\textsf{mKCGEA}$_p$}

\def\ket#1{\mathinner{|{#1}\rangle}}
\def\bra#1{\mathinner{\langle{#1}|}}

\begin{document}
\definecolor{mycommentcolor}{RGB}{0,150,0}

\title{Cryptanalysis of Isogeny-Based Quantum Money with Rational Points}

\author{Hyeonhak Kim}
\orcid{0000-0003-2782-6999}
\affiliation{%
  School of Cybersecurity,
  Korea University,
  Seoul,
  Republic of Korea
}
\email{gusgkr0117@korea.ac.kr}

\author{Donghoe Heo}
\orcid{0000-0001-9300-2985}
\affiliation{%
  Agency for Defense Development,
  Seoul,
  Republic of Korea
}
\email{dhheo.res@gmail.com}

\author{Seokhie Hong}
\orcid{0000-0001-7506-4023}
\affiliation{%
  SmartM2M,
  Busan,
  Republic of Korea
}
\email{shhong@smartm2m.kr}

\author{Sangjin Lee}
\affiliation{%
 School of Cybersecurity,
  Korea University,
  Seoul,
  Republic of Korea
}
\email{sangjin@korea.ac.kr}

\author{Suhri Kim}
\orcid{0000-0003-2665-7142}
\affiliation{%
 School of Mathematics, Statistics and Data Science,
 Sungshin Women's University,
 Seoul,
 Republic of Korea
}
\email[Corresponding author: ]{suhrikim@sungshin.ac.kr}

\maketitle

\begin{abstract}
  Quantum money is a representative application in cryptography that utilizes quantum technology to achieve security that is unattainable with classical systems. Since the theoretical proposal by Wiesner in 1983, the most recent instantiation is the quantum money using class group actions on elliptic curves, proposed by Montgomery and Sharif in Asiacrypt'24. In this work, we propose a concrete cryptanalysis by leveraging the efficiency of evaluating division polynomials with the coordinates of rational points, offering a speedup of $O(\log^4p)$ compared to the brute-force attack. While the resulting attack still runs in exponential time, our method provides the first nontrivial algorithmic insight into their structural hardness. Furthermore, due to the algebraic properties of rational points and quadratic twists, our technique naturally yields a more efficient verification procedure for the underlying quantum money. We believe this new line of analysis opens promising directions for future research in elliptic-curve-based quantum cryptography.
\end{abstract}

\section{Introduction}
Since Wiesner's original theoretical proposal in 1983 \cite{wiesner1983conjugate}, quantum money has become a compelling direction of research within cryptography. Unlike many earlier applications of quantum technologies, where the notion of quantum advantage primarily referred to improved attack efficiency against classical cryptosystems, quantum money shifts focus toward leveraging quantum algorithms to achieve security that classical systems cannot provide.

Although blockchain-based cryptocurrencies such as Bitcoin and Ethereum also aim to realize secure digital money, their background and motivation substantially differ from that of quantum money. The components of cryptocurrencies are distributed ledgers and consensus mechanisms, and they mainly focus on building a distributed system. In contrast, quantum money basically presumes a centralized model where a trusted issuing authority generates banknotes, and primarily concerns the unforgeability of physical currency.


The security of quantum money is rooted directly in the laws of quantum mechanics. The security of quantum money fundamentally relies on the  \textbf{quantum no-cloning theorem}, which states that it is impossible to create an identical copy of an arbitrary and unknown quantum state. By issuing money encoded as unique quantum states, one can build currencies that are inherently unforgeable. In contrast to traditional banknotes, which ultimately depend on classical anti-counterfeiting technologies, quantum money offers information-theoretic protection against counterfeiting.

Wiesner's initial concept of quantum money was privately verifiable, which means that only the issuing authority (the mint) can verify the money. That was impractical for use in the real world. In \cite{aaronson2009quantum}, the publicly verifiable quantum money scheme was devised, where anyone can verify the money. The first instantiation of publicly verifiable quantum money with a formal security proof \cite{aaronson2012quantum} was based on hidden subspaces. However, in the later study \cite{conde2019non}, it has been shown that there exists a polynomial-time attack for the hidden subspaces problem. Zhandry proposed a construction of quantum money/lightning based on the assumption of indistinguishability obfuscation in \cite{zhandry2021quantum}, which was later broken by Roberts \cite{roberts2021security}. Recently, Zhandry showed a construction of quantum money \cite{zhandry2025quantum} from regular group actions, and Liu, Montgomery and Zhandry devised a construction of quantum money which is called \textit{walkable invariant money} in \cite{liu2023another} based on \cite{zhandry2025quantum}. Subsequently, a quantum money using class group actions on elliptic curves was proposed in \cite{montgomery2025quantum} by Montgomery and Sharif. Their scheme is based on simpler assumptions than that of Zhandry's construction \cite{zhandry2025quantum}, \cite{liu2023another}. This scheme of Montgomery and Sharif is the focus of this paper, and for convenience, we refer to it as \emph{\msqm\ scheme}.

In \msqm, a quantum banknote consists of a quantum state $\ket{\psi_\sigma}$ together with a serial number $\sigma \in \mathbb{Z}$. For a base prime $p$, the quantum state $\ket{\psi_\sigma}$ is defined as the uniform superposition over ordinary elliptic curves defined over $\mathbb{F}_{p}$ with cardinality $\sigma$. The security of \msqm\ relies on the hardness of uniformly sampling ordinary elliptic curves with a prescribed cardinality. From a quantum perspective, this corresponds to the difficulty of producing multiple copies of the uniform superposition over all elliptic curves in a given isogeny class over $\mathbb{F}_p$.

Although \msqm\ is based on a relatively simple hardness assumption, the complexity of the underlying attacks is not concretely analyzed, making it unclear how to choose the base prime $p$ with respect to a target security level $\lambda$. In particular, \cite{montgomery2025quantum} does not provide concrete parameter recommendations. 



\subsection{Contributions}
 In this work, we analyze the unforgeability of the quantum money scheme proposed in \cite{montgomery2025quantum} in a concrete manner. We identify two potential approaches to forge a quantum banknote in this scheme.

The first approach involves constructing a uniform superposition of elliptic curves with the given cardinality using the quantum random walk. In order to use the quantum random walk-based method, the index used in generating the quantum state must be discarded, which is known as the Index Erasure Problem. In Section \ref{section_quantum_random_walk_based_forgery_attack}, we explain the difficulty of directly applying Kuperberg's algorithm \cite{kuperberg2011another} to this problem in our setting and show that a direct SEA implementation of the index search is asymptotically less efficient than our cardinality search.

The second approach leverages quantum search techniques to sample elliptic curves of the required cardinality. A straightforward yet computationally expensive method would be to employ the point-counting algorithm as a search oracle. Instead, we propose a more efficient alternative that takes advantage of the fact that verifying the order of rational points is significantly faster than counting rational points.

Furthermore, we note the intrinsic connection between forging a quantum banknote and the verification process. By exploiting this relationship, we outline how our proposed optimization can be applied not only to an attack strategy but also to an improved verification method. Our main contributions are summarized as follows:

\begin{itemize}
    \item \textbf{Optimized Forgery Attack} : We propose a novel attack strategy that efficiently searches for a quantum banknote with a given serial number, outperforming direct brute-force approaches. Our attack method $O(\log^4 p)$ times faster than the brute-force attack and requires $16\lceil\log p\rceil^2$ qubits, which is $\log p$ times fewer than the brute-force method.
    \item \textbf{Concrete Security Estimation} : As our approach is significantly simpler than implementing the point-counting algorithm, we can provide a more precise estimation of the quantum resources required to forge quantum money. We offer a detailed analysis of the quantum resources. The result is shown in Table \ref{table_resource_estimation}.
    \item \textbf{Improved Verification} : We demonstrate how our insights can improve the efficiency of the verification of a given serial number. The process of checking a serial number becomes $O(\log^4 p)$ times faster than the original approach.
    \item \textbf{Utilization of Rational Points} : We show that rational points can be used to verify the cardinality based on the properties of the quadratic twists of elliptic curves. We expect that our method inspires future research.
\end{itemize}

\subsection{Organization of The Paper}
In Section \ref{section_preliminaries}, we introduce the background on elliptic curves and the notions of quantum money and quantum lightning. We also briefly present the quantum lightning scheme proposed in \cite{montgomery2025quantum}. In Section \ref{section_quantum_random_walk_based_forgery_attack}, we discuss why Kuperberg's algorithm does not directly provide the index erasure needed in our setting and compare a direct SEA implementation of the index search with our cardinality search. Section \ref{section_search algorithm to find a quantum money} illustrates our quantum search algorithm and its oracle. In Section \ref{section_lower and upper bound of class numbers}, we show that the lower and upper bound of class numbers that leads to the number of iterations in the quantum search algorithm. In Section \ref{section_implementation of our oracle}, our concrete quantum oracle is described and we analyzed its space and time complexity. In Section \ref{section_faster verification}, we apply our attack method to the verification algorithm. Finally, we conclude our work in Section \ref{section_conclusion}.

\section{Preliminaries}
\label{section_preliminaries}
Throughout this paper, $\log$ denotes $\log_2$, and $\ln$ denotes the natural logarithm.
\subsection{Class Group Actions on Elliptic Curves}
\subsubsection{Elliptic Curves}
\label{subsubsection_elliptic curves}
For a large prime $p$, an elliptic curve $E/\mathbb{F}_p$ is defined as follows in the short Weierstrass form:

$$
E: y^2 = x^3 + Ax + B
$$
where $A, B \in \mathbb{F}_p$. The isomorphism class of elliptic curves can be represented as its $j$-invariant $j(E) = 1728\frac{4A^3}{4A^3+27B^2}.$ Two elliptic curves $E_1, E_2$ are isomorphic over the algebraically closed field $\bar{\mathbb{F}}_p$ if and only if $j(E_1) = j(E_2)$, which means that $j$-invariant uniquely represents the $\bar{\mathbb{F}}_p$-isomorphism classes of elliptic curves. When we consider $\mathbb{F}_p$-isomorphism classes of elliptic curves, we need additional information to uniquely identify the classes. We introduce the representation used in \cite{montgomery2025quantum}. The $\mathbb{F}_p$-isomorphism classes of elliptic curves can be represented by pairs $(j,b) \in \mathbb{F}_p \times \mathbb{Z}$ where $j\in \mathbb{F}_p$ is the $j$-invariant of an elliptic curve defined over $\mathbb{F}_p$, and $b \in \{0,1\}$ except in the following cases:
\begin{itemize}
    \item If $j \equiv 1728 \mod p$ and $p \equiv 1 \mod 4$, then $0 \leq b \leq 3$.
    \item If $j \equiv 0 \mod p$ and $p \equiv 1 \mod 3$, then $0 \leq b \leq 5$.
\end{itemize}
We call $b$ the \emph{twisting data} of a given elliptic curve $E$ defined over $\mathbb{F}_p$. The short Weierstrass form of a curve $E$ is determined by a pair $(j,b)$ as follows.

For a quadratic non-residue $\alpha_2 \in \mathbb{F}_p$, we can recover the Weierstrass pair $(A, B) \in \mathbb{F}_p$ from the given pair $(j, b)$ as follows:\\
If $j \not \equiv 1728, 0 \mod p$,
\begin{equation}
\label{eq: twist equation 1}
y^2 = x^3 + \frac{3j\alpha_2^{2b}}{1728-j}x + \frac{2j\alpha_2^{3b}}{1728-j}
\end{equation}
If $j \equiv 1728 \mod p$, the elliptic curve is given by
\begin{equation}
\label{eq: twist equation 2}
y^2 = x^3 + \alpha_4^bx
\end{equation}
where $\alpha_4$ is a quartic non-residue if $p \equiv 1 \mod 4$, a quadratic non-residue otherwise.\\
If $j \equiv 0 \mod p$,
\begin{equation}
\label{eq: twist equation 3}
y^2 = x^3 + \alpha_6^b
\end{equation}
where $\alpha_6$ is a sextic non-residue if $p \equiv 1 \mod 3$, a quadratic non-residue otherwise.

By \cite[Corollary X.5.4.1]{Silverman:1338326}, there is a one-to-one correspondence between pairs $(j,b)$ and $\mathbb{F}_p$-isomorphism classes of elliptic curves. Although the $(j,b)$ representation is used in \msqm, we mainly work with the Weierstrass representation $(A,B)$ in this paper. Unlike the $j$-invariant representation, Weierstrass pairs can be incorporated more directly into our quantum search algorithm.

According to the above relations (Equation \ref{eq: twist equation 1}, \ref{eq: twist equation 2}, and \ref{eq: twist equation 3}), we convert the $j$-invariant representation into the Weierstrass form. Algorithm~\ref{algorithm_get weierstrass pair} describes the corresponding quantum procedure. The algorithm requires only a small number of multiplications in $\mathbb{F}_p$.

Algorithm~\ref{algorithm_get weierstrass pair} takes as input a prime $p$, a $j$-invariant $j \in \mathbb{F}_p$, and a twisting data $b$, and outputs the quantum state $\ket{A}\ket{B}$ corresponding to the elliptic curve
\[
E_{A,B} : y^2 = x^3 + Ax + B,
\]
such that $j(E_{A,B}) = j$, and where $b$ satisfies one of Equations~\ref{eq: twist equation 1}, \ref{eq: twist equation 2}, or \ref{eq: twist equation 3}.

\begin{remark}
    Although the conversion between $(j,b)$ and short Weierstrass forms $(A,B)$ is classically computable, we describe Algorithm~\ref{algorithm_get weierstrass pair} as a quantum algorithm that takes a quantum state as input and outputs a quantum state. This quantum version is necessary because we need to transform a superposition of $(j,b)$ pairs into the corresponding superposition of Weierstrass pairs $(A,B)$ during our forgery attack.
\end{remark}

\begin{algorithm}[h]
\caption{Algorithm $\mathsf{GetWeierstrassPair}_{p}$}
\label{algorithm_get weierstrass pair}
\SetKwInput{KwInput}{Input}
\SetKwInput{KwOutput}{Output}
\KwInput{A prime $p$ and a quantum state $\ket{j}\ket{b}$ where $(j, b) \in \mathbb{F}_p \times \mathbb{Z}$ and $0\leq b \leq 5$}
\KwOutput{A quantum state $\ket{A}\ket{B}$, where the Weierstrass curve $E_{A,B} : y^2 = x^3 + Ax+B$ has $j$-invariant of $j$ and its twisting data of $b$.}
\uIf{$p \equiv 1 \mod 12$}{
Set $\alpha_2$, $\alpha_4$ and $\alpha_6$ as quadratic, quartic and sextic non-residue respectively.
}
\uElseIf{$p \equiv 1 \mod 4 \text{ and } p \not\equiv 1 \mod 3$}{
Set $\alpha_2$ and $\alpha_6$ as quadratic non-residues and $\alpha_4$ as a quartic non-residue.
}
\uElseIf{$p \equiv 1 \mod 3 \text{ and } p \not\equiv 1 \mod 4$}{
Set $\alpha_2$ and $\alpha_4$ as quadratic non-residues and $\alpha_6$ as a sextic non-residue.
}
\Else {
Set all $\alpha_2$, $\alpha_4$ and $\alpha_6$ as quadratic non-residues.
}
Compute $\ket{A} \leftarrow \ket{(j\equiv1728)\times\alpha_4^b +(j\not\equiv0,1728)\times\left(\frac{3j\alpha_2^{2b}}{1728-j}\right) \mod p}$.\\
Compute $\ket{B} \leftarrow \ket{(j\equiv0)\times\alpha_6^b + (j\not\equiv0,1728)\times\left(\frac{2j\alpha_2^{3b}}{1728-j}\right) \mod p}$.\\
\Return $\ket{A}\ket{B}$.
\end{algorithm}

We also note that the group structure of an elliptic curve defined over $\mathbb{F}_p$ is $\mathbb{Z}/m\mathbb{Z} \times \mathbb{Z}/mk\mathbb{Z}$.

\begin{theorem}
\label{theorem_elliptic curve group structure}
    For a prime $p$, the cardinality of an elliptic curve defined over $\mathbb{F}_p$ is isomorphic to
    $$
    \mathbb{Z}/m\mathbb{Z} \times \mathbb{Z}/mk\mathbb{Z}
    $$
    where $m, k \in \mathbb{Z}^+$ and $m | (p-1)$.
\end{theorem}

\begin{proof}
    We refer to the proof in \cite{ruck1987note}.
\end{proof}

 We note that the number of points of elliptic curves is bounded by the well-known Theorem \ref{theorem_hasse theorem}. We use this property in Section \ref{section_implementation of our oracle} to prove the correctness and the soundness of our algorithm.

\begin{theorem}[Hasse's theorem]
\label{theorem_hasse theorem}
Let $E$ be an elliptic curve defined over a finite field $\mathbb{F}_p$. Then
$$
|\#E(\mathbb{F}_p)-p-1| \leq 2\sqrt{p}
$$
\end{theorem}

\begin{proof}
    We refer to the proof in \cite[Theorem 1.1]{Silverman:1338326}.
\end{proof}

\subsubsection{Frobenius Endomorphism and Counting Points on Elliptic Curves over $\mathbb{F}_p$}
Given an elliptic curve defined over $\mathbb{F}_p$, there exists a polynomial time algorithm to count the number of points on the curve.

\begin{theorem}
\label{thm: schoof's point counting}
Given an elliptic curve $E$ defined over $\mathbb{F}_p$, there is a classical algorithm for computing the number of points defined over $\mathbb{F}_p$ on the curve $E$ in time $O(\log^8 p)$.
\end{theorem}

Schoof's point counting algorithm computes the trace of the \emph{Frobenius endomorphism} $\pi$ which is defined by
\begin{align*}
    \pi : E &\to E\\
    (x,y) &\mapsto (x^p, y^p).
\end{align*}
The Frobenius endomorphism is known to satisfy a quadratic polynomial of the form
$$
    \pi\circ \pi -[t]\circ \pi +[p] = \mathcal{O}_E,
$$
where $[t]$ and $[p]$ are multiplication-by-$t$ and $p$ maps, and $\mathcal{O}_E$ is a constant map that maps all points to the point at infinity. $t$ is called the \emph{Frobenius trace}, and it is known that the number of points $\#E(\mathbb{F}_p)$ is equal to $p+1 -t$. By Theorem \ref{theorem_hasse theorem}, it satisfies that $|t| < 2\sqrt{p}$. 

\subsubsection{Class Group and Frobenius Discriminant}
We define the \emph{Frobenius discriminant} to be
$$
    \Delta_\mathrm{Fr}(E) = 4p-t^2.
$$
Given an ordinary elliptic curve $E$, we denote its endomorphism ring defined over $\mathbb{F}_p$ as $\mathrm{End}_p(E)$. It is an order $\mathcal{O}=\mathbb{Z}+f\mathcal{O}_K$ in $K=\mathbb{Q}(\sqrt{t^2-4p})$, where $f$ is a positive integer and $\mathcal{O}_K$ is the maximal order. The \emph{ideal class group} $\mathrm{Cl}(\mathcal{O})$ of the order $\mathcal{O}$ is defined by
$$
    \mathrm{Cl}(\mathcal{O}):= \{\text{invertible fractional ideals of }\mathcal{O}\}/\{\text{nonzero principal fractional ideals of }\mathcal{O} \}.
$$
The ideal class group $\mathrm{Cl}(\mathcal{O})$ acts on the set $\mathcal{I}_{t}$ which is defined by
$$
    \mathcal{I}_t := \left\{\mathbb{F}_p\text{-isomorphism classes of ordinary elliptic curves } E' \mid \#E'(\mathbb{F}_p)=p+1-t\right\}.
$$
The class group acts simply and transitively on the classes with fixed endomorphism ring $\mathcal{O}$. In \msqm, $4p-t^2$ is square-free, so every curve in $\mathcal{I}_t$ has endomorphism ring $\mathbb{Z}[\pi]=\mathcal{O}_K$. Consequently, $|\mathcal{I}_t|=h(t^2-4p)$.

\subsubsection{Quadratic Twist of Elliptic Curves}
Given an elliptic curve $E : y^2=x^3+Ax+B$ defined over $\mathbb{F}_p$, a quadratic twist of $E$ is defined as $E^t:y^2=x^3+\alpha^{-2}Ax+\alpha^{-3}B$ where $\alpha$ is a quadratic non-residue in $\mathbb{F}_p$. There is an isomorphism $\phi : E \rightarrow E^t$ defined over $\mathbb{F}_{p^2}$ such that $\phi(x,y) = (\alpha^{-1}x,\alpha^{-3/2}y)$. Note that $E$ and $E^t$ are not isomorphic over $\mathbb{F}_p$.

\begin{theorem}
\label{theorem_quadratic twist}
Given an elliptic curve $E : y^2=x^3+Ax+B$ over $\mathbb{F}_p$ and its quadratic twist $E^t$ of $E$, they satisfy the following equation.
$$
\#E(\mathbb{F}_p) + \#E^t(\mathbb{F}_p) = 2p+2.
$$
\end{theorem}

\begin{proof}
    For an $x \in \mathbb{F}_p$, if $x^3+Ax+B$ is a quadratic residue over $\mathbb{F}_p$, there exist $\pm y$ such that $y^2 = x^3+Ax+B$, which means $(x,\pm y) \in E(\mathbb{F}_p)$. Otherwise, there exist $\pm y$ such that $\alpha y^2=x^3+Ax+B$ for a quadratic non-residue $\alpha$, which means $(\alpha^{-1}x,\pm\alpha^{-1}y) \in E^t(\mathbb{F}_p)$. When $y=0$, $(x,0) \in E(\mathbb{F}_p)$ and $(\alpha^{-1}x,0) \in E^t(\mathbb{F}_p)$. It says that the sum of the sizes of two groups $E(\mathbb{F}_p)$ and $E^t(\mathbb{F}_p)$ equals to $2(|\mathbb{F}_p| + 1) = 2p+2$, considering the point at infinity $0_E$ and $0_{E^t}$.
\end{proof}

\subsubsection{Class Group Actions}
For an elliptic curve $E$, we say that $E$ has complex multiplication by $\mathcal{O}$ if $\text{End}(E) \cong \mathcal{O}$ where $\mathcal{O}$ is an order of an imaginary quadratic number field $K$.

According to Deuring correspondence, an element $\alpha \in \mathcal{O}$ corresponds to an endomorphism $\theta_\alpha \in \text{End}(E)$ and an integral ideal $I \subset \mathcal{O}$ corresponds to an isogeny $\phi_I : E \rightarrow E_I := E/E[I]$ where $E[I] =\{\cap \ker(\theta_\alpha): \alpha \in I\}$.

\begin{theorem}[\sea\ isogeny computation]
\label{theorem_time complexity of class group action}
Let $E$ be an elliptic curve over $\mathbb{F}_p$ with complex multiplication $\mathcal{O}$, and let $\mathfrak{l} \subset \mathcal{O}$ be a prime ideal of norm $\ell$, where $\ell$ is a rational prime. Then there is a classical algorithm which computes the isogeny $\varphi_\mathfrak{l}$ in time complexity $O(\ell M(p)\log \ell\log\log\ell\log p)$ where $M(p)$ is the complexity of a multiplication in $\mathbb{F}_p$.
\end{theorem}

We refer to \cite{de2018towards} and \cite{montgomery2025quantum} for the proof and the concrete implementation as a quantum algorithm.

\subsection{Quantum Money and Quantum Lightning}
In this section, we define public key quantum money and quantum lightning. Quantum lightning was first introduced in \cite{lutomirski2009breaking} as a collision-free quantum money protocol, and is named quantum lightning in \cite{zhandry2021quantum}. Both quantum money and quantum lightning consist of two algorithms \textsf{Gen} and \textsf{Ver} as follows:

\begin{itemize}
    \item $\textsf{Gen}(1^\lambda)$. Takes as input a security parameter $\lambda$ and generates a quantum money state $\ket{\psi}$ and the associated serial number $\sigma$.
    \item $\textsf{Ver}(\sigma, \ket{\psi})$. Takes as input a pair of a serial number $\sigma$ and a supposed quantum money $\ket{\psi}$, and verifies them.
\end{itemize}

\begin{definition}[Quantum Money Unforgeability]
$(\textsf{Gen}, \textsf{Ver})$ is secure quantum money if, for all quantum polynomial-time adversaries $\mathcal{A}$, it is negligible probability that $\mathcal{A}$ wins the following game :
\begin{itemize}
    \item The challenger runs $(\ket{\psi}, \sigma) \leftarrow \textsf{Gen}(1^\lambda)$ and give $\sigma, \ket{\psi}$ to $\mathcal{A}$.
    \item $\mathcal{A}$ produces and sends to the challenger two supposed quantum money $\ket{\psi_1}$ and $\ket{\psi_2}$.
    \item The challenger runs $b_1 \leftarrow \textsf{Ver}(\sigma, \ket{\psi_1})$ and $b_2 \leftarrow \textsf{Ver}(\sigma, \ket{\psi_2})$. If $b_1 = b_2 = 1$, $A$ wins.
\end{itemize}
\end{definition}

\begin{definition}[Quantum Lightning Unforgeability]
$(\textsf{Gen}, \textsf{Ver})$ is secure quantum lightning if, for all quantum polynomial-time adversaries $\mathcal{A}$, it is negligible probability that $\mathcal{A}$ wins the following game :
\begin{itemize}
    \item $\mathcal{A}$, on input $1^\lambda$, produces and sends to the challenger a serial number $\sigma$ and supposed quantum money $\ket{\psi_1}$ and $\ket{\psi_2}$.
    \item The challenger runs $b_1 \leftarrow \textsf{Ver}(\sigma, \ket{\psi_1})$ and $b_2 \leftarrow \textsf{Ver}(\sigma, \ket{\psi_2})$. If $b_1 = b_2 = 1$, $\mathcal{A}$ wins.
\end{itemize}
\end{definition}
A quantum money forger also gives a quantum lightning forger: run $\mathsf{Gen}$ to obtain a banknote and its serial number, and then use the money forger to duplicate it. The main difference between quantum money and quantum lightning is whether the issuing authority can duplicate banknotes.

\subsection{Quantum Money from Class Group Actions}
\label{subsection_quantum money from class group actions}
In this section, we introduce \msqm\ in \cite{montgomery2025quantum}, the construction of quantum lightning using class group actions on elliptic curves. $\mathsf{Gen}$ uses $\mathsf{ECSupGen}$ as a subroutine which makes a uniform superposition of elliptic curves over $\mathbb{F}_p$. $\mathsf{Ver}$ uses $\mathsf{ECSupVer}$ as a subroutine that verifies a serial number and the uniformity of a superposition of a quantum banknote. The isogeny computation is based on the \sea\ isogeny algorithm in Theorem \ref{theorem_time complexity of class group action}.

\textsf{ECSupGen} takes as input a base prime $p$, and outputs a uniform superposition $\ket{E}$ of the set of all short Weierstrass forms defined over $\mathbb{F}_p$.

\begin{algorithm}[h]
\caption{Algorithm \textsf{ECSupGen}}
\label{algorithm_ECSupGen}
\SetKwInput{KwInput}{Input}%
\SetKwInput{KwOutput}{Output}%
\KwInput{$p$ a prime}
\KwOutput{$\ket{E}$ a quantum state}
Let $\mathcal{S}$ be a register that can store a pair $(j,b)$, where $j \in \mathbb{F}_p$ and $0 \leq b \leq 5$.\\
Generate a uniform superposition $\ket{\psi} \in \mathcal{S}$ over all pairs $(j,b)$, where
\begin{itemize}
    \item If $j\not \equiv 0, 1728 \mod p$, then $b = 0$ or 1.
    \item If $j\equiv 1728$ and $p \equiv 1 \mod 4$, then $0 \leq b \leq 3$. If $p \equiv 3 \mod 4$, then $b = 0$ or 1.
    \item If $j\equiv 0$ and $p \equiv 1 \mod 3$, then $0 \leq b \leq 5$. If $p \equiv 2 \mod 3$, then $b = 0$ or 1.
\end{itemize}
\end{algorithm}
In \msqm, a quantum banknote $\ket{\psi}$ is a uniform superposition of the set of elliptic curves such that all curves have the same cardinality $\sigma$. The cardinality $\sigma$ is used as the serial number of a quantum banknote. The \textsf{Gen} algorithm (Algorithm \ref{algorithm_mint}) describes the process of creating a quantum banknote $\ket{\psi}$ and an associated serial number $\sigma$. The \textsf{Gen} algorithm takes as input a base prime $p$, and outputs a quantum state $\ket{\psi}$ and a serial number $\sigma \in \mathbb{Z}$. It creates a uniform superposition of elliptic curves defined over $\mathbb{F}_p$ by using the \textsf{ECSupGen} (Algorithm \ref{algorithm_ECSupGen}), and counts the number of rational points of each elliptic curve in the uniform superposition. The point-counting algorithm is Schoof's algorithm in \cite{schoof1995counting}. 

\begin{algorithm}[h]
\caption{Algorithm \textsf{Gen}}
\label{algorithm_mint}
\SetKwInput{KwInput}{Input}%
\SetKwInput{KwOutput}{Output}%
\KwInput{$p$ a prime}
\KwOutput{$\ket{\psi}$ a quantum state and an associated serial number $\sigma \in \mathbb{Z}$}
Compute a superposition $\sum\ket{j,b}$ over all elliptic curves over $\mathbb{F}_p$ by running $\mathsf{ECSupGen}(p)$.\\
Use Schoof's point-counting algorithm to compute the cardinality in $\sum\ket{j,b}$ in superposition and obtain $\sum\ket{j,b}\ket{\#E_{j,b}}$\tcp*{Theorem \ref{thm: schoof's point counting}}
Measure the last register and obtain $\sigma$.\\
Compute $\Delta_{\text{Fr}}(E)=4p-(\sigma-p-1)^2$ and set a third register to be 1 if $\Delta_{\text{Fr}}(E)$ is square-free and $\Delta_{\text{Fr}}(E) > 3p$, and 0 otherwise.\\
Measure the third register; if the result is 0, start over at step 1.\\
\Return The quantum state $\sum_{\#E_{j,b}=\sigma}\ket{j,b}$ and the associated serial number $\sigma$.
\end{algorithm}

The \textsf{ECSupVer} algorithm (Algorithm \ref{algorithm_ECSupVer}) takes as input a pair of a quantum banknote and a serial number, and outputs $0$ or $1$ depending on the validity of the banknote. It uses Schoof's point-counting algorithm (Theorem \ref{thm: schoof's point counting}) to compute the associated serial number.

The uniformity of the quantum money is verified by applying all possible isogenies and checking whether the resulting quantum state remains unchanged. Intuitively, the \textsf{ECSupVer} algorithm simulates a random walk in superposition over the isogeny class group, and repeatedly checks whether the state has changed using projection-valued measurements. If the money state is valid, it remains unchanged throughout this process, since the random walk maps uniform superpositions to uniform superpositions.

In order to compute an isogeny with superposition of elliptic curves, we need to use the \sea\ isogeny algorithm (Theorem \ref{theorem_time complexity of class group action}), specifically Elkies steps (Algorithm 3 and 4) in \cite[p. 12]{de2018towards}. In Algorithm \ref{algorithm_ECSupVer}, we use the following notations.

\begin{itemize}
    \item $\ket{\textbf{1}_n} := \frac{1}{\sqrt{n}}\sum_{i=1}^{n}{\ket{i}}$.
    \item $U := \sum_{i=1}^r{\ket{i}\bra{i}\otimes \sigma_i} + \sum_{i=r+1}^{2r}{\ket{i}\bra{i}\otimes I_k}$ where $r = \#B_K$, $k = \#\mathcal{I}_N$, $\sigma_i$ is a group action in $B_K$ and $\mathcal{I}_N$ is the set of elliptic curves with $N$ points.
    \item The group action $\sigma_i$ can be efficiently evaluated by Theorem \ref{theorem_time complexity of class group action}.
    \item $B_K$ is \textit{Bach Generating Set} generated by \cite[Algorithm 4.1]{montgomery2025quantum}, which is the set of ideal classes of unramified primes $\mathfrak{l}$ of a ring of integers $\mathcal{O}_K$ of an imaginary quadratic field $K$ with $N(\mathfrak{l}) < 6(\log|\text{Disc}(K)|)^2$.
\end{itemize}

\begin{algorithm}[h]
\caption{Algorithm \textsf{ECSupVer}}
\label{algorithm_ECSupVer}
\SetKwInput{KwInput}{Input}%
\SetKwInput{KwOutput}{Output}%
\KwInput{a prime $p$, integers $N$ and $\tau$, and a quantum state $\ket{\psi}$ stored in a register $\mathcal{S}$}
\KwOutput{a bit 0 or 1, then \textsf{ECSupVer} alters $\ket{\psi}$ to a state $\ket{\psi'}$ which it then outputs}
Check that $\ket{\psi}$ is properly formatted as a superposition over pairs $(j,b) \in \mathbb{F}_p \times [0,5]$ with $b$ following the conditions of Algorithm \ref{algorithm_ECSupGen}, and that $4p - (p+1-N)^2 \geq 3p$ and is square-free. If not output 0\\
Use Schoof's algorithm to compute the cardinality of the elliptic curve in a new register.\tcp*{Theorem \ref{thm: schoof's point counting}}
Measure the value in the new register. If it is not $N$, output 0. Otherwise, compute the list of group actions $B_K$ and discard the new register.\\
Let $r = \#B_K$. Using a new register, create the state $\ket{\varphi}:=\ket{\textbf{1}_{2r}}\otimes\ket{\psi}$.\\
\For{$i \leftarrow 1$ \KwTo $\tau$} {
    Apply the unitary $U$ to $\ket{\varphi}$.\tcp*{Theorem \ref{theorem_time complexity of class group action}}
    Apply the projection-valued measurement corresponding to $\ket{\textbf{1}_{2r}}\bra{\textbf{1}_{2r}}$ to the resulting state. If the measurement fails, output 0.
}
\Return 1 and the resulting curve state $\ket{\psi\prime}$.
\end{algorithm}

Using the \textsf{ECSupVer} algorithm as a subroutine, the \textsf{Ver} algorithm (Algorithm \ref{algorithm_Ver}) verifies the given quantum banknote and a serial number.
By \cite[Proposition 8.3]{montgomery2025quantum}, The \textsf{Ver} algorithm runs in 
$$
\max(O(\log^8p), O(\tau(\log^5p)(\log\log^2p)(\log\log\log^2p)))
$$
where $\tau = 33r^3\lambda$ for $r = \#B_K$. The second term corresponds to the running time of the \sea\ isogeny algorithm (Theorem \ref{theorem_time complexity of class group action}), which dominates or is at least comparable to the complexity of the point-counting algorithm when $\lambda = O(\log p)$ and $\#B_K = O(\log p)$.

\begin{algorithm}[h]
\caption{Algorithm \textsf{Ver}}
\label{algorithm_Ver}
\SetKwInput{KwInput}{Input}%
\SetKwInput{KwOutput}{Output}%
\KwInput{a quantum state $\ket{\psi}$ and a serial number $\sigma\in\mathbb N$}
\KwOutput{$(0, \bot)$ or $(1, \ket{\psi'})$}
Run $\mathsf{ECSupVer}(p, N = \sigma, \tau, \ket{\psi})$ and receive an output tuple $(b,\ket{\psi'})$.\\
\If{$b=0$}{\Return $(0, \bot)$ and discard $\ket{\psi'}$.}
\Return 1 and $\ket{\psi'}$.
\end{algorithm}

\subsection{Security Assumptions}
The above isogeny-based quantum money and quantum lightning schemes are based on the hardness assumption of the following problems. Problem \ref{problem: ECSDP} reduces to the quantum money forgery game and Problem \ref{problem: ECSCP} reduces to the quantum lightning forgery game. 

\begin{problem}[\cite{montgomery2025quantum}, Elliptic Curve Superposition Duplication Problem (ECSDP)] 
\label{problem: ECSDP}
The prime $p$ is fixed. Suppose we consider an elliptic curve $E_s$ sampled uniformly at random from all elliptic curves $E$ over $\mathbb{F}_p$ with Frobenius discriminant $\Delta_{\mathrm{Fr}}(E) \geq 3p$ and square-free. Given a superposition $\ket{\rho} = \sum\ket{E}$ of all elliptic curves over $\mathbb{F}_p$ in the same isogeny class as $E_s$, create two (possibly entangled) quantum states that are negligibly far from $\ket{\rho}$.
\end{problem}

\begin{problem}[\cite{montgomery2025quantum}, Elliptic Curve Superposition Collision Problem (ECSCP)]
\label{problem: ECSCP}
The prime $p$ is fixed. Create two (possibly entangled) quantum states that are each negligibly far from the superposition of all elliptic curves over $\mathbb{F}_p$ in some isogeny class with Frobenius discriminant $\Delta_{\mathrm{Fr}}\geq 3p$ and square-free.
\end{problem}

The hardness of ECSCP is also based on the two assumptions: the hardness of the isogeny path finding problem and the knowledge of the path assumption. The isogeny path finding problem is formalized in \cite{montgomery2025quantum} as follows. It asks to find a group action between $E_1$ and $E_2$ for two given curves $E_1$ and $E_2$.

\begin{problem}[Arbitrary Isogeny Class Group Action Discrete Logarithm Problem with prime $p$ (\aicgadlog)]
\label{prob: aicgadlog}
For some prime $p$, the \aicgadlog\ problem is defined between a challenger and an adversary as follows: 
\begin{enumerate}
\item The adversary picks an elliptic curve $E$ over $\mathbb{F}_p$ such that $\Delta_\mathrm{Fr}(E)\geq 3p$.
\item The challenger selects two elliptic curves $E_1$ and $E_2$ uniformly at random from the isogeny class of $E$.
\item The adversary wins the game if it outputs an isogeny that map $E_1$ to $E_2$.
\end{enumerate}
\end{problem}

The knowledge of path assumption states that the oblivious sampling of an elliptic curve over a finite field $\mathbb{F}_p$ is hard. This is formalized as the quantum modified knowledge of class group element assumption (\qmkcgea) (\cite{montgomery2025quantum}, \cite{zhandry2025quantum}). \qmkcgea\ is a modified version of the knowledge of group element assumption (KGEA), which uses the following notion of generic group action games.

\begin{definition}[Isogeny Class Group Action Game]
    A \emph{one-round isogeny class group action game} with parameter $p$ is given by an interactive algorithm called a challenger, which we will denote by $\mathsf{Ch}$. In the game, $\mathsf{Ch}$ sends a set of (classical) parameters derived from some setup algorithm taking as input a security parameter $\lambda$  (and $p$) to an adversary $\mathcal{A}$ and then receives back a single quantum message from $\mathcal{A}$ before deciding if $\mathcal{A}$ wins, and outputs 1 if $\mathcal{A}$ wins. We denote such a game $\mathsf{Ch}^p$ and say that an adversary $\delta(\lambda)$-wins $\mathsf{Ch}^p$ if $\mathsf{Ch}^p(\lambda)$ outputs 1 with probability at least $\delta(\lambda)$ when interacting with $\mathcal{A}$.
\end{definition}

Based on the group action games, \qmkcgea\ is defined as follows.

\begin{definition}[Quantum Modified Knowledge of Class Group Element Assumption with parameter $p$ (\qmkcgea)]
\label{def: qmkcgea}
For some prime $p$, the \qmkcgea\ assumption holds if the following is true:

Consider a one-round isogeny class group action game $\mathsf{Ch}^p$ and any QPT algorithm $\mathcal{A}$ that $\delta(\lambda)$ wins $\mathsf{Ch}^p$ for some non-negligible $\delta$.

Suppose we write the message from $\mathcal{A}$ to $\mathsf{ch}$ as $\rho_{1,2,3}$, which we define as a joint system over three registers, which we denote 1, 2, and 3, respectively. Consider measuring the first and second registers and suppose that we obtain elliptic curves $E_1$ and $E_2$, where $E_1$ is the starting element, and $E_2$ is a set element that passes the class group action game $\mathsf{Ch}^p$. We denote this process and output as $(E_1E_2\ket{\psi})\leftarrow \mathcal{A}(1^\lambda) \leftrightarrow \mathsf{Ch}^p(1^\lambda)$. 

Then, for all such $\delta, \mathcal{A}$ and $\mathsf{Ch}^p$, there exists another non-negligible $\delta'$ and a QPT $\mathcal{A}'$ that $\delta'$-wins $\mathsf{Ch}^p$, and moreover there exists a QPT extractor $\mathcal{E}$ and a negligible function $\epsilon$ of $\lambda$ such that the following holds:
$$
    \Pr\left[E_1, E_2 \in \mathscr{E}_{\mathsf{sf},3p} \text{ and } E_2 \neq \varphi(E_1) : \substack{(E_1E_2\ket{\psi}) \leftarrow \mathcal{A}'(1^\lambda) \leftrightarrow \mathsf{Ch}^p(1^\lambda)\\ \varphi \leftarrow \mathcal{E}(E_1 E_2 \ket{\psi})}\right] \leq \epsilon(\lambda),
$$
where $\mathscr{E}_{\mathsf{sf},3p}$ denotes the set of elliptic curves $E$ defined over $\mathbb{F}_p$ such that $\Delta_\mathsf{Fr}(E) \geq 3p$ and it is square-free. 
\end{definition}

\begin{theorem}[Theorem 9.11 in \cite{montgomery2025quantum}]
    Suppose that the \aicgadlog\ problem (Problem \ref{prob: aicgadlog}) is hard and that \qmkcgea\ assumption (Definition \ref{def: qmkcgea}) holds true. Then, the ECSCP$_p$ problem (Problem \ref{problem: ECSCP}) is hard.
\end{theorem}

\section{Quantum Random Walk Based Forgery Attack}
\label{section_quantum_random_walk_based_forgery_attack}
In this section, we introduce an attack method based on the quantum random walk and compare a direct SEA implementation of its index search with our cardinality-search attack.
\subsection{Difficulty of Applying Kuperberg's Algorithm}
\label{section_difficulty of applying kuperberg algorithm}
Given a serial number $\sigma$, there is no known general polynomial-time algorithm to construct an elliptic curve with exact cardinality $\sigma$. However, suppose that one manages to obtain such a curve $E_\sigma$, for instance, by collapsing and measuring a given quantum money state. From $E_\sigma$ as the starting curve, one can simulate a quantum random walk over the isogeny graph, resulting in an (almost) uniform superposition of the form $\sum_{\mathfrak{a}}\ket{\mathfrak{a}}\ket{[\mathfrak{a}]E_\sigma}$. To pass the verification process, we must discard the first register $\ket{\mathfrak{a}}$ and retain only the second. This task is known as the Index Erasure Problem.

It is well known that the discrete logarithm problem in group actions can be reduced to the HSP (Hidden Shift Problem) and it can be solved in sub-exponential time by using Kuperberg's algorithm \cite{kuperberg2011another}. However, applying Kuperberg's algorithm directly to index erasure requires preserving a superposition over target curves, which is not provided by its standard hidden-shift formulation.

In Kuperberg's algorithm, a weak Fourier measurement is applied to a quantum query of the hidden function $f$, which yields a qubit (known as a phase vector) whose phases encode information about the hidden shift $s$:
$$\sum_{0\leq j<\ell}\mathrm{exp}(2\pi i b_j s/2^n)\ket{j}$$
This procedure uses measurements for a fixed hidden-shift instance. In our setting, we need to preserve a superposition over multiple instances, so the standard procedure does not directly perform the required index erasure. We therefore consider quantum search to recover the index $\mathfrak{a}$, which requires exponential time in this approach.

\subsection{Time Complexity of Solving the Index Erasure Problem}

By using quantum search, we aim to find an ideal class $[\mathfrak a]$ such that $[\mathfrak a]E_\sigma=E$. Put $n=\lceil\log p\rceil$ and $r=\#B_K$. Assuming GRH, Bach's result guarantees that the prime ideal classes with norm below $6(\log|d|)^2$ generate $\mathrm{Cl}(\mathcal O_K)$ \cite[p.376]{bach1990explicit}. As in \cite[Proposition 4.17]{montgomery2025quantum}, we use the upper bound $r=O((\log|d|)^2)=O(n^2)$.

We also use the weakly Erd\H{o}s--R\'enyi assumption of \cite[Section 4.4]{montgomery2025quantum}: writing $B_K=\{[\mathfrak l_1],\ldots,[\mathfrak l_r]\}$, every ideal class has a representation
\[
[\mathfrak a]=\prod_{i=1}^{r}[\mathfrak l_i]^{e_i},\qquad e_i\in\{0,1\}.
\]
Thus large exponents are unnecessary: evaluating a given representation requires at most $r$ generator actions. Since the $2^r$ exponent choices cover all $h(d)$ classes, we also have $r\geq\log h(d)$.

Consider the direct implementation that evaluates this product by controlled SEA isogeny computations for the generators in $B_K$. To compare the search costs, suppose the search register labels the $h(d)$ ideal classes once each and supplies a short representation for each label. Denote the cost of evaluating the action and reversing its work registers by $T_1$. Grover search for the index then uses order $\sqrt{h(d)}$ action evaluations. Our attack uses order $\sqrt{2^{n+3}/h(d)}$ evaluations of an oracle of cost $T_2=\Theta(n^2)$ field multiplications, as estimated in Section~\ref{section_implementation of our oracle}. The ratio of the total arithmetic costs is therefore, up to constant factors,
\begin{equation}
\label{eq:comparison of index search and cardinality search}
\frac{T_1\sqrt{h(d)}}{T_2\sqrt{2^{n+3}/h(d)}}
=\frac{h(d)}{\sqrt{2^{n+3}}}\frac{T_1}{T_2}.
\end{equation}

We count $T_1$ for the usual modular-polynomial root computation used in the Elkies steps of \cite[Algorithms 3 and 4]{de2018towards}. For a generator of norm $\ell_i=N(\mathfrak l_i)$, this includes computing
\[
\gcd\bigl(\Phi_{\ell_i}(X,j),X^p-X\bigr).
\]
Computing $X^p$ modulo the degree-$(\ell_i+1)$ polynomial by repeated squaring takes $\Theta(n)$ modular polynomial squarings. With standard polynomial arithmetic that stores all coefficients, each such squaring uses at least order $\ell_i$ field multiplications. The direct controlled circuit includes these routines for each generator, giving
\[
T_1=\Omega\!\left(n\sum_{i=1}^{r}\ell_i\right).
\]
This counts the root-computation routines of this implementation; the remaining isogeny computations and their reversal can only add to the cost.

Each rational prime is the norm of at most two of the prime ideals representing $B_K$. Since the $k$th rational prime has size $\Theta(k\log k)$ by the prime number theorem, we obtain
\[
\sum_{i=1}^{r}\ell_i=\Omega(r^2\log r),\qquad
T_1=\Omega(nr^2\log r).
\]
Whenever the class-number lower bound in Section~\ref{section_lower and upper bound of class numbers} applies, $h(d)=\Omega(\sqrt p/n)$ and hence $r\geq\log h(d)=\Omega(n)$. Consequently $T_1=\Omega(n^3\log n)$. Substituting this and $T_2=\Theta(n^2)$ into Equation~\ref{eq:comparison of index search and cardinality search}, and using $2^{n+3}=\Theta(p)$, gives
\[
\frac{T_1\sqrt{h(d)}}{T_2\sqrt{2^{n+3}/h(d)}}
=\Omega\!\left(\frac{h(d)}{\sqrt p}\,n\log n\right)
=\Omega(\log n).
\]
Thus, in this comparison the cost of evaluating the generator actions outweighs the possible saving in Grover iterations.

The direct SEA implementation also requires more qubits. A residue modulo the degree-$(\ell_i+1)$ polynomial $\Phi_{\ell_i}(X,j)$ has $\ell_i+1$ field coefficients, requiring $(\ell_i+1)n$ qubits when all coefficients are stored. The same prime-counting argument gives $\max_i\ell_i=\Omega(r\log r)$; with $r=\Omega(n)$, even one such polynomial register therefore requires $\Omega(n^2\log n)$ qubits. This is asymptotically larger than the approximately $16n^2$ qubits used by our algorithm in Section~\ref{section_implementation of our oracle}.

As the quantum random walk-based attack is asymptotically less efficient and requires more qubits with this direct SEA implementation, we focus on optimizing the quantum search algorithm for curves of the required cardinality.

\section{General Search Algorithm To Forge Quantum Money}
\label{section_search algorithm to find a quantum money}

In order to find a quantum banknote associated to the given serial number $\sigma$, we use Grover search algorithm. Grover search algorithm consists of the oracle $O_{p, \sigma}$, $n$-bit Walsh-Hadamard gate $W_n$ and the phase rotation gate $O_0 = -2\ket{0^{n+3}}\bra{0^{n+3}}+I$. The search algorithm requires $O(1/\sqrt{r})$ queries to $O_{p,\sigma}$ to obtain a uniform superposition of the target set $T$ where $r = |T|/|X|$. Figure \ref{fig:grover search algorithm} represents the general Grover search algorithm.

\begin{figure}[h]
    \centering
    \includegraphics[width=0.8\linewidth]{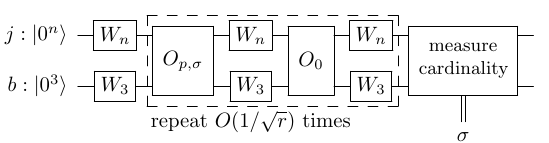}
    \caption{Grover search algorithm to forge a quantum money}
    \label{fig:grover search algorithm}
\end{figure}
Put $n=\lceil\log p\rceil$. The Hadamard gates in Figure~\ref{fig:grover search algorithm} prepare all $2^{n+3}$ bit strings. The oracle leaves invalid encodings unchanged and marks only valid target curves. Thus $|X|=2^{n+3}$ and $|T|=h(d)$, giving $O(\sqrt{2^{n+3}/h(d)})$ iterations.

The general oracle $O_{p, \sigma}$ runs as follows:
\begin{itemize}
    \item The oracle is classically initialized by a positive integer $\sigma$ satisfying $0<|\sigma-p-1|\leq2\sqrt{p}$ which represents the cardinality of the elliptic curves to be sampled.
    \item It takes as input a pair of elliptic curve coefficient $(j,b) \in \mathbb{F}_p \times \mathbb{Z}$ and the size of the pair of the quantum registers is $\lceil \log p\rceil + 3$.
    \item Oracle flips the phase of the quantum state if the Weierstrass curve $E_{j,b}$ has cardinality $\sigma$.  
\end{itemize}

The direct approach to get the desired quantum banknote is to search with the minting algorithm illustrated in Algorithm \ref{algorithm_mint}. Since the minting algorithm uses Schoof's point-counting algorithm, search oracle $O_{p,\sigma}$ is to count the number of rational points and check if it is the same with the given $\sigma$. Because the point-counting algorithm requires to compute arithmetics on a polynomial ring of a large degree over a finite field, it requires $O(\log^8 p)$ bit operations in total as mentioned in \cite{schoof1995counting}. In Section \ref{section_implementation of our oracle}, we show that we can construct a more optimal oracle than using the point-counting algorithm. 

To predict the overall time complexity of the search algorithm, we need to compute the ratio of solutions that pass through the oracle $O_{p, \sigma}$. The more accurately we compute the ratio of the solutions, the more efficiently the search algorithm can be executed. In our case, the size of the target set $T$ equals to the class number of a imaginary quadratic field $K$ such that the endomorphism ring of the elliptic curves are isomorphic to the maximal order $\mathcal{O}_K$ of $K$ since the quantum state generated by Algorithm \ref{algorithm_mint} satisfies that $\Delta_{\text{Fr}}(E)$ is square-free and $\mathbb{Z}[\text{Fr}]=\text{End}(E) \cong \mathcal{O}_K$.

Computing the class number of a given number field requires sub-exponential time in a classical setting \cite[Section 5.4]{cohen2013course}. However, a polynomial-time quantum algorithm has been developed to compute the class number \cite{biasse2016efficient}. This result is based on a quantum reduction from the class group problem (CGP) to the continuous hidden subgroup problem (CHSP), which can be efficiently solved in a quantum setting. The cost of solving CHSP has been tightly estimated in \cite{de2020quantum}.

Although we do not delve into the quantum algorithm for CHSP in this paper, the class number is crucial to determine the number of search iterations. We provide the lower and upper bound of class numbers, which leads to the bound of the number of search iterations.
\section{Lower and Upper Bound of Class Numbers}
\label{section_lower and upper bound of class numbers}
For the elliptic curves generated by Algorithm \ref{algorithm_mint}, the number of $\mathbb{F}_p$-isomorphism classes of elliptic curves of cardinality $\sigma$ is equal to the class number $h(d)$, where $d$ is the discriminant of an imaginary quadratic order $\mathcal{O}_K$ isomorphic to $\text{End}(E)$. In \cite{montgomery2025quantum}, the lower bound of the class number was provided. The lower bound is calculated from the Dirichlet class number formula and the bound of the Dirichlet series $L(1, \chi)$.

By Dirichlet class number formula, given a fundamental discriminant $d < -4$, the class number $h(d)$ of the imaginary quadratic field of discriminant $d$ satisfies the following equation.

$$
h(d) = \frac{\sqrt{|d|}}{\pi}L(1,\chi)
$$
where $L(1,\chi)$ is Dirichlet $L$-function and $\chi(m) = \left(\frac{d}{m}\right)$ is Kronecker symbol. 

\begin{theorem}
    Fix $0 < \epsilon < \frac{1}{2}$. For all negative fundamental discriminants $d$ with $|d| \geq \max(e^{1/\epsilon}, e^{11.2})$, except possibly one, and $\chi(m) = \left(\frac{d}{m}\right)$, we have
    $$
        L(1, \chi) > 0.655\frac{\epsilon}{|d|^\epsilon}
    $$
\end{theorem}
\begin{proof}
    We refer to the proof in \cite[Theorem 2]{tatuzawa1952theorem}.
\end{proof}

For $p>2^{20}$, choose $\epsilon = 1/\ln p$. Since $3p\leq |d|\leq4p$, the theorem gives the following lower bound, apart from the possible exception above.

$$
h(d) = \frac{\sqrt{|d|}}{\pi}L(1,\chi)
>\frac{0.655\sqrt{3p}}{\pi\ln p\,(4p)^{1/\ln p}}
> 0.08\frac{\sqrt{p}}{\log p}.
$$
Using a fundamental result in number theory, we can also derive the upper bound of class numbers. Next, we show that the upper bound of class numbers can be approximated based on Pólya-Vinogradov inequality.

\begin{theorem}[Pólya-Vinogradov Inequality]
    Let $d>1$ be a positive integer and $\chi(k)$ a nonprincipal Dirichlet character modulo $d$. Then
    $$
        \forall m\geq n,\quad \left|\sum_{k=n}^{m}\chi(k)\right| = O(\sqrt{d}\log{d}).
    $$
\end{theorem}
Pólya-Vinogradov Inequality is improved in \cite{pomerance2011remarks}. They suggest a more explicit version of inequality for the sum of values of a Dirichlet character on an interval.

\begin{theorem}[Theorem 1, \cite{pomerance2011remarks}]
    \label{theorem_explicit version of polya-vinogradov}
    Let $d>1$ be an integer and $\chi(k)$ a primitive Dirichlet character modulo $d$. Then 
    $$
        \forall m\geq n,\quad \left|\sum_{k=n}^{m}\chi(k)\right| \leq \begin{cases}
            d^{1/2}\left(\frac{2}{\pi^2}\ln d + \frac{4}{\pi^2}\ln{\ln d}+\frac{3}{2}\right) & \text{if } \chi \text{ is even.}\\
            d^{1/2}\left(\frac{1}{2\pi}\ln d + \frac{1}{\pi}\ln{\ln d}+1\right) & \text{if } \chi \text{ is odd.}
        \end{cases}.
    $$
\end{theorem}

\begin{proof}
    The proof is in \cite[Theorem 1]{pomerance2011remarks}.
\end{proof}

\begin{theorem}
\label{theorem_upperbound of L function}
    For a negative fundamental discriminant $d<-4$ and $\chi(m) = \left(\frac{d}{m}\right)$ modulo $|d|$,
    $$
    L(1,\chi) \leq \left(\frac{1}{2}+\frac{1}{2\pi}\right)\ln{|d|} + \frac{1}{\pi}\ln\ln{|d|} + 1
    $$
\end{theorem}
\begin{proof}
    Partial summation gives
    \begin{align*}
    L(1, \chi) &=\int_{1}^{\infty}{\left(\sum_{n \leq x}{\chi(n)}\right)x^{-2}dx}.
    \end{align*}
    Since $d$ is negative, $\chi(-1)=-1$ which is odd Dirichlet character. By Theorem \ref{theorem_explicit version of polya-vinogradov},\\
    \begin{align*}
    &\int_{1}^{\infty}{\left(\sum_{n \leq x}{\chi(n)}\right)x^{-2}dx}\leq \int_{1}^{\sqrt{|d|}}{x \cdot x^{-2}dx} +\\
    &\int_{\sqrt{|d|}}^{\infty}{\left(\frac{1}{2\pi}|d|^{1/2}\ln{|d|} + \frac{1}{\pi}|d|^{1/2}\ln{\ln {|d|}}+|d|^{1/2}\right) \cdot x^{-2}dx}\\
    &= \left(\frac{1}{2}+\frac{1}{2\pi}\right)\ln{|d|} + \frac{1}{\pi}\ln\ln{|d|} + 1\\
    \end{align*}
\end{proof}

Theorem \ref{theorem_upperbound of L function} directly leads to the upper bound of the class number. As $|d|=|4p-t^2| \leq 4p$,

\begin{align}
\label{equation_upperbound of class number}
h(d) = \frac{\sqrt{|d|}}{\pi}L(1,\chi) \leq \left(\frac{1+\pi}{\pi^2}\right)\sqrt{p}\ln{(4p)} + \frac{2}{\pi^2}\sqrt{p}\ln\ln{(4p)} + \frac{2}{\pi}\sqrt{p}.
\end{align}
We can see that the result of the calculating class number $h(d)$ using \cite{biasse2016efficient} is bounded by the above inequality (\ref{equation_upperbound of class number}). The ratio $\sqrt{2p/h(d)}$ is bounded as follows; the right inequality excludes the possible exception in Tatuzawa's theorem:

\begin{equation}
\label{eq: lower and upper bound of grover iterations}
\frac{\sqrt{2}\pi p^{1/4}}{\sqrt{(\pi+1)\ln{(4p)}+2\ln\ln{(4p)}+2\pi}} \leq \sqrt{\frac{2p}{h(d)}} \leq \sqrt{\frac{2}{0.08}}p^{1/4}\sqrt{\log{p}}.
\end{equation}
The lower bound is approximately $2.622\frac{p^{1/4}}{\sqrt{\log{p} + \Theta(\log\log{p})}}$.
For the full $n+3$ input bits in Figure~\ref{fig:grover search algorithm}, where $n=\lceil\log p\rceil$, multiply both sides by $\sqrt{2^{n+2}/p}$ to obtain bounds on $\sqrt{2^{n+3}/h(d)}$.

\section{Implementation of Our Attack}
In this section, we explain about our Grover search oracle and illustrate how division polynomials are used in our attack algorithm. The aim of the search oracle is to verify the number of points in a given elliptic curve. For efficiency, we get around counting the points in an elliptic curve by checking the order of arbitrary chosen rational points. One may think of the Montgomery ladder as a scalar multiplication algorithm. However, we adopt using division polynomials, which does not require any auxiliary operations. This is also explained in Remark \ref{remark: why we use division polynomials}.
\label{section_implementation of our oracle}
\subsection{Division Polynomials and Recurrence Relation}
\label{subsubsection_Division Polynomials and Recurrence Relation}

Given a prime $p$ and a positive integer $\ell$, one can calculate a division polynomial $\psi_{\ell}(x,y) \in \mathbb{F}_p[x,y]$ such that $\psi_{\ell}(x, y) = 0$ if and only if $(x,y) \in E[\ell]$. For a given Weierstrass form $E: y^2 = x^3 + Ax + B$ where $A, B \in \mathbb{F}_p$, we define the \textit{division polynomials} as follows :

\begin{align*}
\psi_{-1} &= -1,\\
\psi_0 &= 0,\\
\psi_1 &= 1,\\
\psi_2 &= 2y,\\
\psi_3 &= 3x^4 + 6Ax^2 + 12Bx - A^2\\
\psi_4 &= 2y(2x^6 + 10Ax^4 + 40Bx^3 - 10A^2x^2 - 8ABx - 2A^3 - 16B^2),\\
\psi_{2n+1} &= \psi_{n+2}\psi_{n}^3-\psi_{n-1}\psi_{n+1}^3 = g_1(\psi_{n-1}, \psi_{n}, \psi_{n+1}, \psi_{n+2}) &\text{for } n \geq 2,\\
\psi_{2n} &= \frac{\psi_{n-1}^2\psi_n \psi_{n+2} - \psi_{n-2}\psi_n \psi_{n+1}^2}{\psi_2} = g_2(\psi_{n-2}, \psi_{n-1}, \psi_{n}, \psi_{n+1}, \psi_{n+2}) &\text{for } n \geq 3
\end{align*}

To avoid computing a square root or dividing by $2y$, we evaluate the normalized polynomials $D_m=\psi_m$ for odd $m$ and $D_m=\psi_m/(2y)$ for even $m$, where the division is performed symbolically before evaluation. Put $v=x^3+Ax+B$ and $c=16v^2$. Then $D_0=0$, $D_1=D_2=1$, $D_3=\psi_3$, and $D_4$ is the expression in parentheses in $\psi_4$. Each $D_m$ is a single element of $\mathbb{F}_p$ after evaluation. The normalized recurrences are
\begin{align*}
D_{2m}&=D_m(D_{m+2}D_{m-1}^2-D_{m-2}D_{m+1}^2),\qquad m\ge2,\\
D_{2m+1}&=\begin{cases}
cD_{m+2}D_m^3-D_{m-1}D_{m+1}^3,&m\text{ even},\\
D_{m+2}D_m^3-cD_{m-1}D_{m+1}^3,&m\text{ odd},
\end{cases}\qquad m\ge2.
\end{align*}
For $v\ne0$, $D_m(A,B,x)=0$ if and only if $[m](x,y)=0_E$. For $v=0$, the point $(x,0)$ has order two, so it is annihilated precisely when $m$ is even. We handle this case separately below.

Given a tuple $\Psi_k := (D_k,\ldots,D_{k+9})\in\mathbb{F}_p^{10}$, we obtain $\Psi_{2k+4}$ or $\Psi_{2k+5}$ using the normalized recurrences. For $m > 5$, $\Psi_m$ can be computed as follows:

$$
\Psi_m = \begin{cases}
f_1(\Psi_{(m - 4)/2}) & \text{if $m$ is even}\\
f_2(\Psi_{(m - 5)/2}) & \text{if $m$ is odd}
\end{cases}
$$
where $g_2$ denotes the right-hand side of the normalized even recurrence and $g_{1,m}$ that of the normalized odd recurrence, including the known parity of $m$:

\begin{align}
\label{equation_f1 and f2}
f_1(\Psi_k)[i] &= 
    \begin{cases}
    g_2(\Psi_k[i/2], \ldots, \Psi_k[i/2+4]) & \text{if } i \text{ is even} \\
    g_{1,k+(i-1)/2+2}(\Psi_k[(i-1)/2+1], \ldots, \Psi_k[(i-1)/2+4]) & \text{if } i \text{ is odd}
    \end{cases} \\
f_2(\Psi_k)[i] &= 
    \begin{cases}
    g_{1,k+i/2+2}(\Psi_k[i/2+1], \ldots, \Psi_k[i/2+4]) & \text{if } i \text{ is even} \\
    g_2(\Psi_k[(i-1)/2 + 1], \ldots, \Psi_k[(i-1)/2+5]) & \text{if } i \text{ is odd}
    \end{cases}\notag
\end{align}
for $0 \leq i < 10$, with $\Psi_k[i]=D_{k+i}$. Each evaluation retains its input, computes its output in a fresh register, and resets temporary products by reversing their computation. Algorithm \ref{algorithm_division polynomial} is the quantum algorithm that computes division polynomials.

\begin{algorithm}
\caption{Algorithm $\mathsf{DivisionPolynomial_{p, \sigma}}$}
\label{algorithm_division polynomial}
\SetKwInput{KwInput}{Input}%
\SetKwInput{KwOutput}{Output}%
\KwInput{A prime $p$, a cardinality $\sigma \in \mathbb{N}$ and a quantum state $\ket{A}\ket{B}\ket{x}$ where $A, B, x \in \mathbb{F}_p$}
\KwOutput{$\ket{D_\sigma(A,B,x)}$, retaining the inputs and resetting all work registers}
Compute $v$, $c$, and $\ket{\Psi_{\sigma_r}}$ from $\ket{A}\ket{B}\ket{x}$.\\
\For{$i\leftarrow 0$ \KwTo $r-1$} {
Compute $\ket{\Psi_{\sigma_{r-i-1}}} \leftarrow \ket{f_{b_i}(\Psi_{\sigma_{r-i}})}$.
}
Copy $D_\sigma=\Psi_{\sigma_0}[0]$ into an initially zero output register using controlled-NOT gates.\\
\For{$i=r-1,r-2,\ldots,0$} {
Reverse the evaluation of $f_{b_i}$ using the retained $\Psi_{\sigma_{r-i}}$, resetting the register holding $\Psi_{\sigma_{r-i-1}}$ to zero.
}
Uncompute $\ket{\Psi_{\sigma_r}}$, $c$, and $v$ from $\ket{A}\ket{B}\ket{x}$.\\
\Return $\ket{D_\sigma(A,B,x)}$.
\end{algorithm}
In Algorithm \ref{algorithm_division polynomial}, $f_1$ and $f_2$ are the functions in Equation (\ref{equation_f1 and f2}) and $\sigma_i$'s and $b_i$'s are precomputed as follows:
\begin{align*}
\sigma_0 &= \sigma \text{ and } \sigma_r \leq 5\\
\sigma_{i+1} &= \begin{cases}
(\sigma_i -4)/2 & \text{if $\sigma_i$ is even and }\sigma_i > 5\\
(\sigma_i -5)/2 & \text{if $\sigma_i$ is odd and }\sigma_i > 5\\
\sigma_i & \text{if $\sigma_i \leq 5$}
\end{cases}\\
b_{r-i-1} &= \begin{cases}
1 & \text{if $\sigma_{i}$ is even}\\
2 & \text{if $\sigma_{i}$ is odd}
\end{cases}
\end{align*}
where $r$ is the smallest nonnegative integer such that $\sigma_r = \sigma_{r+1}$ and $0 \leq i < r$. Both loops are empty when $r=0$.

Each evaluation of $f_1$ or $f_2$ computes five odd and five even entries. The normalized odd recurrence uses at most seven field multiplications or squarings, and the even recurrence at most five, giving at most $5\times7+5\times5=60$ operations. Computing these products, copying the ten outputs, and reversing the products requires at most 120 field-multiplication circuit calls. Since Algorithm~\ref{algorithm_division polynomial} performs approximately $2\log\sigma$ such evaluations, its cost is approximately $240\log\sigma$ field multiplications, or $240\log^2p\log\sigma$ bit operations when a field multiplication is approximated by $(\log p)^2$ bit operations.


\subsection{Our Search Oracle with Rational Points}

In this section, we describe our quantum search oracle $O_{p,\sigma}$ used to forge a quantum banknote in Figure \ref{fig:grover search algorithm}. The oracle $O_{p, \sigma}$ consists of the evaluation of division polynomials and the phase rotation via the $n$-controlled-$Z$ gate $Z_n$. In order to filter out all but the target curves, we verify that all rational points are annihilated by $\sigma$, while all rational points on its quadratic twist are annihilated by $2p+2-\sigma$. Given $x \in \mathbb{F}_p$, we determine the appropriate case by checking whether $x^3+Ax+B$ is a quadratic residue. If it is, we verify annihilation by $\sigma$; otherwise, we verify annihilation by $2p+2-\sigma$. We then define $G_{p, \sigma}$ as a function that calculates the division polynomial with respect to the original curve and its twist. Instead of applying scalar multiplications, we compute division polynomials to verify annihilation, as this approach is computationally more efficient.

\begin{remark}
\label{remark: why we use division polynomials}
Using the Montgomery ladder algorithm, scalar multiplication can be performed more efficiently than with division polynomials. Since the Montgomery ladder operates only on $x$-coordinates, it also avoids square-root computations during scalar multiplication. However, to use the Montgomery ladder, one must first derive a Montgomery coefficient from the given $j$-invariant. For a Montgomery curve
\[
    E_A: y^2 = x^3 + Ax^2 + x,
\]
we call $A$ the Montgomery coefficient, and its $j$-invariant is given by
\[
    j(E_A)=\frac{256(A^2-3)^3}{A^2-4}.
\]
Consequently, recovering the corresponding Montgomery coefficient requires solving a cubic equation together with a square-root computation. This overhead makes the use of the Montgomery ladder slower than using division polynomials.
\end{remark}

For elliptic curves $E_{A,B}$ such that $\#E_{A,B}\neq \sigma$, some rational points may still be annihilated by $\sigma$. The set of such points forms a subgroup $S\subset E_{A,B}(\mathbb{F}_p)$. If $S$ is a proper subgroup, its complement $E_{A,B}\setminus S$ consists of the union of all other cosets of $S$, ensuring that $|E_{A,B}\setminus S| \geq |S|$. Consequently, if at least one rational point is not annihilated by $\sigma$, then at least half of the rational points remain unannihilated. We exploit this fact to distinguish the target curve from the others by computing $F_{p,\sigma,\tau}$.

For a positive integer $\tau$, we define $F_{p,\sigma,\tau}$ as follows.
Choose independent uniform coordinates $\boldsymbol{x}=(x_0,\ldots,x_{\tau-1})\in\mathbb{F}_p^\tau$. Write $\mathbf1[\mathcal C]=1$ when $\mathcal C$ holds and zero otherwise. We use
\begin{align*}
G_{p,\sigma}(A,B,x)&=\begin{cases}
\mathbf1[\sigma\text{ odd}],&v=0,\\
\mathbf1[D_\sigma(A,B,x)\ne0],&v^{(p-1)/2}=1,\\
\mathbf1[D_{2p+2-\sigma}(A,B,x)\ne0],&v^{(p-1)/2}=-1,
\end{cases}\quad v=x^3+Ax+B,\\
F_{p,\sigma,\tau}(A,B;\boldsymbol{x})&=\sum_{i=0}^{\tau-1}G_{p,\sigma}(A,B,x_i).
\end{align*}
The sum is an integer between $0$ and $\tau$, so it is zero exactly when all tests pass. It is stored in the $n$-qubit register $r$ by addition modulo $2^n$; $\tau=3n<2^n$ for the primes considered here. In the nonsquare case, a point over $\mathbb{F}_{p^2}$ corresponds to a rational point of the twist, so the same normalized polynomials can be evaluated at the original $A,B,x$.
The function $F_{p,\sigma, \tau}$ corresponds to Algorithm \ref{algorithm_F}. $\mathsf{DivisionPolynomial_{p,\sigma}}$ refers to Algorithm \ref{algorithm_division polynomial}.
\begin{algorithm}[h]
\caption{Algorithm $F_{p,\sigma,\tau}$}
\label{algorithm_F}
\SetKwInput{KwInput}{Input}%
\SetKwInput{KwOutput}{Output}%
\KwInput{A prime $p$, a cardinality $\sigma$, a fixed classical list $\boldsymbol{x}$, and registers $\ket A\ket B\ket{0^n}_x\ket{0^n}_r$, with $0 < |\sigma-p-1|\leq2\sqrt p$}
\KwOutput{$\ket{F_{p,\sigma,\tau}(A,B;\boldsymbol{x})}_r$, with the inputs retained and other work registers reset}
\For{$i\leftarrow 0$ \KwTo $\tau-1$} {
Load $x_i$ into register $x$, compute $v=x_i^3+Ax_i+B$, and compute Euler's criterion $\ket{t}\leftarrow\ket{v^{(p-1)/2}}$.\\
Compute the division polynomials $\ket{r_{i,1}} \leftarrow \mathsf{DivisionPolynomial}_{p,\sigma}(\ket{A}\ket{B}\ket{x})$ and $\ket{r_{i,2}} \leftarrow \mathsf{DivisionPolynomial}_{p,2p+2-\sigma}(\ket{A}\ket{B}\ket{x})$.\\
Compute the bit $G_{p,\sigma}(A,B,x_i)$ using $v,t,r_{i,1},r_{i,2}$, and add it to register $r$.\\
Reverse the bit computation, uncompute $t$ and $v$, and reset register $x$ to zero.
}
\For{$i=\tau-1,\tau-2,\ldots,0$} {
Load $x_i$ into register $x$.\\
Uncompute the division polynomials $\ket{r_{i,1}} \leftarrow \mathsf{DivisionPolynomial}_{p,\sigma}^{-1}(\ket{A}\ket{B}\ket{x})$ and $\ket{r_{i,2}} \leftarrow \mathsf{DivisionPolynomial}_{p,2p+2-\sigma}^{-1}(\ket{A}\ket{B}\ket{x})$ and discard $\ket{r_{i,1}}$ and $\ket{r_{i,2}}$.\\
Reset register $x$ to zero.
}
\Return $\ket{r}$
\end{algorithm}

Here we show that Algorithm \ref{algorithm_F} runs correctly. The proof is based on Theorem \ref{theorem_no exception case} which is also known as Mestre's theorem.

\begin{theorem}
\label{theorem_no exception case}
    Given a large prime $p > 2^{20}$ and an positive integer $\sigma$ such that $0<|\sigma - p-1| \leq 2\sqrt{p}$, there is no elliptic curve $E_{A,B}$ defined over $\mathbb{F}_p$ such that $[\sigma]P=0_{E}$ for all rational points $P \in E_{A,B}(\mathbb{F}_p)$, $[2p+2-\sigma]Q=0_{E^t}$ for all rational points $Q \in E^t_{A,B}(\mathbb{F}_p)$ and the cardinality of $E_{A,B}$ differs from $\sigma$, where $E^t$ is a quadratic twist of $E$.
\end{theorem}

\begin{proof}
Suppose that a curve $E_{A,B}$ satisfies the condition. According to Theorem \ref{theorem_elliptic curve group structure}, let's say that $m_1, k_1, m_2, k_2$ are positive integers such that
\begin{align*}
    E_{A,B} &\cong \mathbb{Z}/m_1\mathbb{Z} \times \mathbb{Z}/m_1k_1\mathbb{Z}\\
    E^t_{A,B} &\cong \mathbb{Z}/m_2\mathbb{Z} \times \mathbb{Z}/m_2k_2\mathbb{Z}
\end{align*}
where $E^t_{A,B}$ is a quadratic twist of $E_{A,B}$ and $m_1, m_2 | (p-1)$.
Let's denote the cardinality $\#E_{A,B}$ as $\sigma'$. Since $m_1^2k_1 = \sigma'$ and $m_2^2k_2 = 2p+2-\sigma'$ by Theorem \ref{theorem_quadratic twist}, we get
\begin{align}
    \label{equation_1}
    \gcd(m_1^2k_1, m_2^2k_2) &= \gcd(\sigma', 2p+2-\sigma')\\
    &= \gcd(\sigma', 2p+2)\notag
\end{align}
According to the condition, all elements in $\mathbb{Z}/m_1k_1\mathbb{Z}$ are annihilated by both $\sigma$ and $\sigma'$. Thus $m_1k_1| (\sigma - \sigma')$ and likewise, we deduce that $m_2k_2|(2p+2-\sigma - 2p -2 +\sigma')=(\sigma'-\sigma)$. By Hasse's theorem, we obtain
\begin{align}
    \label{equation_2}
    \frac{m_1k_1 \times m_2k_2}{\gcd(m_1k_1,m_2k_2)} \bigg\rvert |\sigma-\sigma'| &\leq 4\sqrt{p}
\end{align}
Since $m_1^2k_1^2 \geq m_1^2k_1 = \#E_{A,B}(\mathbb{F}_p)\geq p+1-2\sqrt{p}$, we get $m_1k_1 \geq \sqrt{p} -1$. Likewise, $m_2k_2 \geq \sqrt{p} -1$ on the twist. As $\sigma - \sigma' \neq 0$,
\begin{align*}
    \gcd(m_1k_1, m_2k_2) \geq \frac{m_1k_1 \times m_2k_2}{4\sqrt{p}} &\geq \frac{(\sqrt{p}-1)^2}{4\sqrt{p}} > \frac{\sqrt{p}}{4} -1
\end{align*}
Put $d_0=\gcd(m_1k_1,m_2k_2)$ and $g=\gcd(k_1,k_2)$. Comparing prime-power divisors gives $d_0/g\mid m_1m_2$. Also $d_0\mid2(p+1)$ by (\ref{equation_1}). Since $m_1,m_2\mid p-1$, we obtain
\begin{align*}
&d_0/g\mid\gcd((p-1)^2,2(p+1))=4,\\
\Rightarrow &k_1,k_2\geq g\geq d_0/4>\frac{\sqrt p}{16}-\frac14.
\end{align*}
Again by Hasse's theorem, $m_1^2k_1 = \#E_{A,B}(\mathbb{F}_p) \leq p+1+2\sqrt{p}$ and
\begin{align*}
    m_1 &\leq \frac{\sqrt{p}+1}{\sqrt{k_1}} < 4\frac{\sqrt{p}+1}{\sqrt{\sqrt{p}-4}} < 5p^{1/4}\\
    \Rightarrow m_1k_1 &= \frac{\#E_{A,B}(\mathbb{F}_p)}{m_1} \geq \frac{p+1-2\sqrt{p}}{m_1} > \frac{p^{3/4}}{6}.
\end{align*}
Likewise, $m_2k_2>p^{3/4}/6$. Using (\ref{equation_2}) again, we get
\begin{align*}
    \gcd(m_1k_1, m_2k_2) &> \frac{m_1k_1 \times m_2k_2}{4\sqrt{p}} > \frac{p}{144}\\
    \Rightarrow \frac{p}{144} < \gcd(m_1^2k_1, m_2^2k_2) &\leq 2\gcd(p+1,\sigma').
\end{align*}

If $\sigma' \neq p+1$, then $\gcd(p+1,\sigma') \leq |\sigma' -p -1| \leq 2\sqrt{p}$, contradicting $p/144>4\sqrt p$ for $p>2^{20}$. Otherwise $\sigma'=p+1$. Since $m_1\mid p-1$ and $m_1\mid p+1$, we have $m_1\leq2$, so $(p+1)/2\mid m_1k_1\mid\sigma$. The only positive multiple of $(p+1)/2$ in the Hasse interval is $p+1$, which is excluded.
\end{proof}
Based on Theorem \ref{theorem_no exception case}, $F_{p,\sigma, \tau}$ can be used to filter out all but the target curves of cardinality $\sigma$. We conclude by Corollary \ref{corollary_correctness of F}.

\begin{corollary}
\label{corollary_correctness of F}
For the prime $p$ and cardinality $\sigma$ in Theorem~\ref{theorem_no exception case}, choose the coordinates in $\boldsymbol{x}$ independently and uniformly. If $\#E_{A,B}(\mathbb F_p)=\sigma$, then $F_{p,\sigma,\tau}(A,B;\boldsymbol{x})=0$ for every list. Otherwise,
\[
\Pr[F_{p,\sigma,\tau}(A,B;\boldsymbol{x})=0]\leq q_p^\tau,
\qquad q_p=\frac34+\frac{1}{2\sqrt p}-\frac{1}{4p}.
\]
\end{corollary}
\begin{proof}
The target curve and its twist pass every test. For a curve of different cardinality, Theorem~\ref{theorem_no exception case} guarantees that at least one of the two groups has a point that fails its annihilation test. The points that pass form a proper subgroup, so at least half of that group's points fail. These are finite points, since the identity always passes. Each $x$-coordinate represents at most two points. By Hasse's theorem, at least $(p+1-2\sqrt p)/4$ coordinates therefore fail. For the twist, multiplication of the coordinate by the fixed nonsquare defining the twist is a bijection of $\mathbb F_p$, so the same bound applies. A uniform coordinate passes with probability at most $q_p$. Since the coordinates are independent and $F$ is zero exactly when every test passes, the bound is $q_p^\tau$.
\end{proof}

Using Algorithm \ref{algorithm_F} as a building block, we replace the search oracle $O_{p,\sigma}$ in Figure \ref{fig:grover search algorithm} by Algorithm \ref{algorithm_our search oracle}. $Z_n$ is the $n$-controlled Pauli-$Z$ gate combined with $n$-bit NOT gates, which flips the phase only if it takes as input $\ket{0^n}$.

\begin{algorithm}
\caption{Algorithm \textsf{OurOracle} $O_{p, \sigma}$}
\label{algorithm_our search oracle}
\SetCommentSty{mycommentstyle}%
\SetKwInput{KwInput}{Input}%
\SetKwInput{KwOutput}{Output}%
\KwInput{A prime $p>2^{20}$, a cardinality $\sigma$ with $0<|\sigma-p-1|\leq2\sqrt p$, a fixed classical list $\boldsymbol{x}$, and an $n+3$ qubit state $\ket j\ket b$}
\KwOutput{The input state with its phase flipped when the encoding is valid and $F_{p,\sigma,\tau}=0$}
Set $n \leftarrow \lceil \log p\rceil$ and $\tau \leftarrow 3\lceil \log p\rceil$.\\
Reversibly check $j<p$ and the valid values of $b$ in Algorithm~\ref{algorithm_ECSupVer}. Apply the following steps only to valid encodings, leaving invalid encodings unchanged.\\
Compute $\ket{A}\ket{B} \leftarrow \mathsf{GetWeierstrassPair}_p(\ket{j}\ket{b})$\\
Set $\ket{x}\ket r\leftarrow\ket{0^n}\ket{0^n}$.\\
Compute $\ket{r} \leftarrow F_{p,\sigma,\tau}(\ket{A}\ket{B};\boldsymbol{x})$.\\
Compute $\ket{r} \leftarrow Z_{n}(\ket{r})$. \tcp{Flip the phase only if $r=0^n$}
Uncompute $\ket{r} \leftarrow F_{p,\sigma,\tau}^{-1}(\ket{A}\ket{B};\boldsymbol{x})$ and discard $\ket{x}$ and $\ket{r}$.\\
Uncompute $\ket{A}\ket{B} \leftarrow \mathsf{GetWeierstrassPair}_p^{-1}(\ket{j}\ket{b})$.\\
Uncompute the validity flag and discard the reset $\ket{A}\ket{B}$ registers.\\
\Return $\ket{j}\ket{b}$.
\end{algorithm}

\begin{figure}[h]
    \centering
    \includegraphics[width=0.8\linewidth]{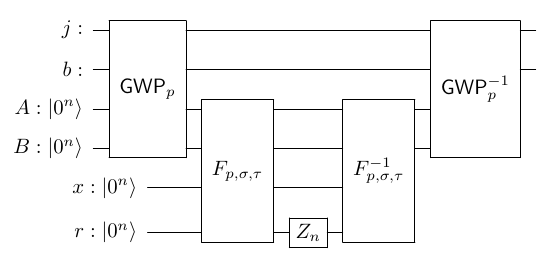}
    \caption{The quantum circuit of \textsf{OurOracle} $O_{p, \sigma}$ used in Figure \ref{fig:grover search algorithm} $\mathsf{GWP}_p$ gate represents the $\mathsf{GetWeierstrassPair}_p$ algorithm and $Z_n$ flips the phase only if $r=\ket{0^n}$.}
    \label{fig:the quantum circuit of the oracle of the division polynomial}
\end{figure}

Choose $\boldsymbol{x}$ once, independently and uniformly, before the Grover search, and reuse this same classical list in every oracle call.
\begin{theorem}
\label{theorem_correctness of our oracle}
For $p>2^{20}$ and $0<|\sigma-p-1|\leq2\sqrt p$, Algorithm~\ref{algorithm_our search oracle} flips exactly the target phases on every valid encoding, except with probability at most $3p q_p^{3\lceil\log p\rceil}$ over the choice of $\boldsymbol{x}$.
\end{theorem}
\begin{proof}
Every target curve passes. By Corollary~\ref{corollary_correctness of F}, each other curve passes with probability at most $q_p^{3\lceil\log p\rceil}$. There are fewer than $3p$ valid encodings, so a union bound gives the stated probability. Moreover,
\[
3p q_p^{3\lceil\log p\rceil}=O\!\left(p^{1-3\log(4/3)}\right)=O(p^{-0.2451\ldots}),
\]
which is negligible in the bit length of $p$.
\end{proof}

\subsection{Quantum Resource Estimation}

Given a serial number $\sigma$ such that $0<|\sigma-p-1|\leq2\sqrt p$, Algorithm~\ref{algorithm_division polynomial} costs approximately $240\log\sigma\approx240\log p$ field multiplications, as derived in Section~\ref{subsubsection_Division Polynomials and Recurrence Relation}. Algorithm~\ref{algorithm_F} computes each of the two division-polynomial values and later uncomputes it for every coordinate. Computing Euler's criterion, copying its value, and resetting its products costs at most $4\log p$ field-multiplication circuit calls; its inverse has the same cost. Thus the full oracle, including $F$ and its inverse, costs approximately
\[
4\tau\bigl(240\log\sigma+240\log(2p+2-\sigma)+4\log p\bigr)
\approx5808\log^2p
\]
field multiplications, using $\tau=3\lceil\log p\rceil\approx3\log p$. Approximating one field multiplication by $(\log p)^2$ bit operations gives approximately $5808\log^4p$ bit operations per iteration. These estimates count the leading multiplication terms; initialization, additions, comparisons, and register control are not included in the coefficient.

\setlength{\tabcolsep}{5pt}
\renewcommand{\arraystretch}{1.2}
\begin{table*}[h]
\centering
\begin{tabular}{|c|ccc|}
\hline
attack method & oracle $O_{p,\sigma}$ & num. of iterations & num. of qubits \\ \hline
brute-force   & $O(\log^6{p})$ $\mathbb{F}_p$-Mul          & $\sqrt{2^{n+3}/h(d)}$   & $O(\log^3{p})$      \\
our method    & \ $\approx 5808\log^2 p$ $\mathbb{F}_p$-Mul           & $\sqrt{2^{n+3}/h(d)}$  & $\ \approx 16\lceil\log{p}\rceil^2$      \\ \hline
\end{tabular}
\caption{The comparison between the brute-force attack and our method. $h(d)$ is the class number corresponding to the given serial number $\sigma$. The boundary of the number of iterations $\sqrt{2^{n+3}/h(d)}$ is calculated in Section \ref{section_lower and upper bound of class numbers}.}
\label{table_resource_estimation}
\end{table*}
From the perspective of the number of qubits, our algorithm also consumes less number of qubits. Each $\Psi_{\sigma_i}$ accounts for $10\lceil\log p\rceil$ qubits. Thus, $10\lceil\log p\rceil^2$ qubits are needed in total. Algorithm \ref{algorithm_F} stores $r_{i,1}$ and $r_{i,2}$ for all $\tau$ coordinates, requiring $2\tau n\approx6\lceil\log p\rceil^2$ additional qubits. Together with the tuple registers, the leading storage is approximately $16\lceil\log p\rceil^2$ qubits. On the other hand, in Schoof's point counting algorithm \cite{schoof1995counting}, the ring element in $\mathbb{F}_p[X, Y]/(\psi_\ell(X), Y^2 - X^3 - AX - B)$ has size $\lceil\log{p}\rceil^3$.

In total, the leading estimate for forging a banknote is $5808\log^4p\sqrt{2^{n+3}/h(d)}$ bit operations, up to the constant factor in the Grover iteration count. Combining the class-number bounds in Equation~\ref{eq: lower and upper bound of grover iterations}, this estimate lies approximately between
\[
\sqrt{\frac{2^{n+2}}p}\,15230\frac{p^{1/4}\log^4p}{\sqrt{\log p+\Theta(\log\log p)}}
\quad\text{and}\quad
\sqrt{\frac{2^{n+2}}p}\,29040p^{1/4}\log^{9/2}p.
\]
The upper estimate excludes the possible exception in the class-number lower bound. The factor $\sqrt{2^{n+2}/p}$ accounts for the full bit-string search space in Figure~\ref{fig:grover search algorithm}. Ignoring constant factors, the relation
\[
p\approx2^{4\lambda}(4\lambda)^{-14}
\]
makes the lower estimate scale as $2^\lambda$. For example, with $p=2^{386}+235$, a 387-bit prime, $\tau=1161$. Substitution in the class-number bounds gives a leading attack-cost estimate between approximately $2^{141.95}$ and $2^{151.49}$ bit operations; this is an estimate of this attack's cost.

\section{Faster Verification of Quantum Money}
\label{section_faster verification}
In the quantum money scheme \cite{montgomery2025quantum}, the verification process is inherently related to the forgery attack. The verification consists of two processes: 1. checking the serial number and 2. checking the uniformity of the quantum state. Each of the two phases is related to two different types of forgery attacks: 1. searching by the serial number and 2. making a superposition by an isogeny walk. Since our attack is the prior case, our method can be applied to checking the validity of the serial number. 
\begin{algorithm}[h]
\caption{Algorithm \textsf{CheckSerialNumber}}
\label{algorithm_our verification}
\SetKwInput{KwInput}{Input}%
\SetKwInput{KwOutput}{Output}%
\KwInput{A prime $p$, a serial number $\sigma \in \mathbb{N}$ and a quantum state $\ket{j}\ket{b}$ where $(j, b) \in \mathbb{F}_p\times\mathbb{Z}$}
\KwOutput{A validity bit, retaining the surviving curve state on acceptance}
Set $n\leftarrow\lceil\log p\rceil$ and $\tau\leftarrow3\lceil\log p\rceil$. Check the encoding and serial-number conditions in Algorithm~\ref{algorithm_ECSupVer}; return 0 if they fail.\\
Choose a fresh list $\boldsymbol{x}$ of $\tau$ independent uniform coordinates, independently of the input state.\\
Compute $\ket{A}\ket{B} \leftarrow \textsf{GetWeierstrassPair}_p(\ket{j}\ket{b})$.\\
Set $\ket{x}\ket r\leftarrow\ket{0^n}\ket{0^n}$.\\
Compute $\ket{r}\leftarrow F_{p,\sigma,\tau}(\ket A\ket B;\boldsymbol{x})$.\\
Set an initially zero qubit $\mathit{acc}$ to 1 if and only if register $r$ equals $0^n$.\\
Apply $F_{p,\sigma,\tau}^{-1}$ to reset registers $r$ and $x$.\\
Uncompute $\ket{A}\ket{B} \leftarrow \textsf{GetWeierstrassPair}_p^{-1}(\ket{j}\ket{b})$\\
\If{Measurement of $\mathit{acc}$ is 1} {\Return 1 and retain the curve state}
\Return 0
\end{algorithm}

We can perform a faster verification using rational points. Our method applies only to checking the validity of the serial number $\sigma$; the verification of the uniformity of the given quantum money remains the same as in \cite{montgomery2025quantum}. The verification algorithm is directly derived from Algorithm~\ref{algorithm_F}. The subroutine $\textsf{GetWeierstrassPair}_p$ converts a $j$-invariant representation $(j,b)$ into the corresponding Weierstrass pair $(A,B)$ using the procedure described in Algorithm~\ref{algorithm_get weierstrass pair}.

The serial-number check is asymptotically $O(\log^4p)$ times faster than checking the serial number by point counting. The uniformity check is unchanged. By Corollary~\ref{corollary_correctness of F}, a curve of incorrect cardinality passes with probability at most
\[
\epsilon_p=q_p^{3\lceil\log p\rceil}\approx p^{-1.2451\ldots},
\]
where $q_p=3/4+1/(2\sqrt p)-1/(4p)$. This is negligible in the bit length of $p$. For an arbitrary input state independent of the fresh list, the expected squared weight of incorrectly accepted components is at most $\epsilon_p$. After the work registers are reset, the accepted state before normalization differs from that of an exact serial-number check by at most $\sqrt{\epsilon_p}$ in expected norm. Applying the original uniformity test preserves this negligible error in the verification output, including its acceptance probability.

\section{Conclusion}
\label{section_conclusion}
In this work, we propose an attack method to forge quantum money in the isogeny-based quantum money scheme presented in \cite{montgomery2025quantum} and introduce a more efficient verification algorithm. We employ the fact that checking the exponent of an elliptic curve group with rational points is more efficient than directly computing its cardinality. Compared to brute-force attack using the point-counting algorithm, our method achieves a speedup of $O(\log^4{p})$. More concretely, we estimate that forging quantum money using our approach requires approximately $5808 \log^2p$ multiplications in $\mathbb{F}_p$ for each search iteration and requires approximately $16\lceil\log p\rceil^2$ qubits.

Our key insight is to utilize rational points on elliptic curves $E_{A,B}$ for the efficient computation of division polynomials. Specifically, our method exploits the property of the group structure of quadratic twists of elliptic curves. As our approach leverages the properties of quadratic twists to utilize rational points, we expect it to contribute to future research on quantum algorithms for elliptic-curve-based quantum cryptography.
\section*{Acknowledgements}
This work was supported by the Institute of Information \& Communications Technology Planning \& Evaluation (IITP) grant funded by the Korea government (MSIT) (No.~RS-2025-02263143, Development of Cybersecurity Threat Response Technologies for Satellite Ground Stations).

\section*{Author contributions}
\noindent H.K. conceived the study, developed the methods, carried out the main research, and wrote the manuscript. D.H. reviewed the initial ideas and discussed directions for their development. S.H., S.L., and S.K. reviewed the technical details of the developed methods.

\paragraph{AI-assisted preparation.}
ChatGPT (OpenAI) was used to assist with calculations, proof checking, and grammar checking.

\bibliographystyle{quantum}
\bibliography{references}

\end{document}